\documentclass[english,draftcls,onecolumn]{IEEEtran}
 
\makeatletter

\usepackage[T1]{fontenc}
\usepackage[utf8]{inputenc}
\usepackage{amsthm}
\usepackage{amsmath}
\usepackage{amssymb}
\usepackage{graphicx}
\usepackage{dsfont}
\usepackage[noadjust]{cite}
\usepackage{amsmath}
\usepackage{array}
\usepackage{multirow}
\usepackage{caption}
\usepackage{color}

\usepackage{multicol}

\theoremstyle{plain}

\theoremstyle{plain}

\theoremstyle{plain}
\newtheorem{lem}{\protect\lemmaname}
\theoremstyle{plain}
\newtheorem{thm}{\protect\theoremname}
\theoremstyle{plain}
  
\theoremstyle{definition}

\theoremstyle{definition}

\theoremstyle{definition}
\newtheorem{rem}{\protect\remarkname}

\makeatother
  
\usepackage{babel} 

\providecommand{\claimname}{Claim}
\providecommand{\lemmaname}{Lemma}
\providecommand{\propositionname}{Proposition}
\providecommand{\theoremname}{Theorem}
\providecommand{\corollaryname}{Corollary} 
\providecommand{\definitionname}{Definition}
\providecommand{\assumptionname}{Assumption}
\providecommand{\remarkname}{Remark}

\newcommand{\overbar}[1]{\mkern 1.25mu\overline{\mkern-1.25mu#1\mkern-0.25mu}\mkern 0.25mu}

\newcommand{\ttil}{\tilde{t}}
\newcommand{\Ttil}{\widetilde{T}}

\newcommand{\Shat}{\widehat{S}}

\newcommand{\kbar}{\overbar{k}}

\newcommand{\Bernoulli}{\mathrm{Bernoulli}}

\newcommand{\pe}{P_{\mathrm{e}}}

\newcommand{\Yv}{\mathbf{Y}}

\newcommand{\Ec}{\mathcal{E}}

\newcommand{\EE}{\mathbb{E}}
\newcommand{\PP}{\mathbb{P}}

\newcommand{\Lc}{\mathcal{L}}

\usepackage{hyperref} 
\usepackage{enumitem}
\usepackage{float}
\usepackage{setspace}

\usepackage{algorithm}
\usepackage{algpseudocode}
\usepackage{chngcntr}

\setenumerate[1]{label=\arabic*.}

\usepackage{geometry} 
\makeatletter
\newcommand{\manuallabel}[2]{\def\@currentlabel{#2}\label{#1}}
\makeatother

\newcommand{\poly}{{\rm poly}}

\begin{document} 

\title{Noisy Adaptive Group Testing \\ via Noisy Binary Search}
\author{Bernard Teo and Jonathan Scarlett}
\date{}
\maketitle

\begin{abstract}
    The group testing problem consists of determining a small set of defective items from a larger set of items based on a number of possibly-noisy tests, and has numerous practical applications.  One of the defining features of group testing is whether the tests are adaptive (i.e., a given test can be chosen based on all previous outcomes) or non-adaptive (i.e., all tests must be chosen in advance).  In this paper, building on the success of binary splitting techniques in noiseless group testing (Hwang, 1972), we introduce noisy group testing algorithms that apply noisy binary search as a subroutine.  We provide three variations of this approach with increasing complexity, culminating in an algorithm that succeeds using a number of tests that matches the best known previously (Scarlett, 2019), while overcoming fundamental practical limitations of the existing approach, and more precisely capturing the dependence of the number of tests on the error probability.  We provide numerical experiments demonstrating that adaptive group testing strategies based on noisy binary search can be highly effective in practice, using significantly fewer tests compared to state-of-the-art non-adaptive strategies.
\end{abstract}

\long\def\symbolfootnote[#1]#2{\begingroup\def\thefootnote{\fnsymbol{footnote}}\footnote[#1]{#2}\endgroup}

\symbolfootnote[0]{ The authors are with the  Department of Computer Science and the Department of Mathematics, National University of Singapore  (e-mail: \url{bernardteo@u.nus.edu}, \url{scarlett@comp.nus.edu.sg}). J.~Scarlett is also with the Institute of Data Science, National University of Singapore.
    
This work was supported by the Singapore National Research Foundation (NRF) under grant number R-252-000-A74-281.}

%
%
\section{Introduction}

The group testing problem consists of determining a small subset $S$ of defective items within a larger set of items $\{1,\dotsc,n\}$, based on a number of possibly-noisy tests. This problem has a history in medical testing \cite{Dor43}, and has regained significant attention following new applications in areas such as communication protocols \cite{Ant11}, pattern matching \cite{Cli10}, database systems \cite{Cor05}, and COVID-19 testing \cite{Ald21}, as well as connections with compressive sensing \cite{Gil08}. In the noiseless setting, each test takes the form
\begin{equation}
    Y = \bigvee_{j \in S} X_j, \label{eq:gt_noiseless_model}
\end{equation}
where the test vector $X = (X_1,\dotsc,X_n) \in \{0,1\}^n$ indicates which items are included in the test, and $Y$ is the resulting observation.  That is, the output indicates whether at least one defective item was included in the test.   One wishes to design a sequence of tests $X^{(1)},\dotsc,X^{(t)}$, with $t$ ideally as small as possible, such that the outcomes can be used to reliably recover the defective set $S$ with probability close to one.

One of the defining features of the group testing problem is the distinction between the {\em non-adaptive} and {\em adaptive} settings.  In the non-adaptive setting, every test must be designed prior to observing any outcomes, whereas in the adaptive setting, a given test $X^{(i)}$ can be designed based on the previous outcomes $Y^{(1)},\dotsc,Y^{(i-1)}$.  It is of considerable interest to determine the extent to which this additional freedom helps in reducing the number of tests.

In the noiseless setting, the problem of finding a near-optimal adaptive algorithm was solved long ago by Hwang \cite{Hwa72}, who proposed an algorithm based on binary splitting (see Section \ref{sec:related}) and showed that it is guaranteed to succeed with $t = k\log_2\frac{n}{k} + O(k)$ tests.  Due to a well-known counting argument, this is asymptotically optimal whenever $k = o(n)$.  In addition, Hwang's algorithm has the advantage of degrading gracefully (i.e., still successfully identifying many defectives) when $t$ falls below the given threshold \cite{Tru20}.

Although perhaps not as ubiquitous as the noiseless version, noisy binary search algorithms have also attracted significant research attention outside the context of group testing \cite{Fei94,Kar07,Ben08,Now09}.  Given the importance of Hwang's noiseless group testing algorithm that uses binary splitting as a subroutine, it is therefore natural to ask the following: {\em Do there exist near-optimal noisy adaptive group testing strategies based on noisy binary search?}  In this paper, we partially answer this question in the affirmative, with near-optimality holding at least in certain scaling regimes, and the number of tests matching the best previously known practical algorithm more generally.  We will see that depending on the desired recovery guarantees and number of tests, somewhat more care is needed in the testing strategy compared to the noiseless setting.

\subsection{Problem Setup} \label{sec:setup}

We let the defective set $S$ be a fixed but unknown subset of $\{1,\dotsc,n\}$ of cardinality $k$.  Throughout the paper, we adopt the widely-used assumption $k = o(n)$ as $n \to \infty$, which is convenient for absorbing certain ``nuisance'' terms into the lower-order asymptotic terms.  More care would be needed with these terms when handling the regime $k = \Theta(n)$ (see \cite{Ald19a} for the analog in the noiseless setting), and this is left for future work.

An adaptive algorithm iteratively designs a sequence of tests $X^{(1)},\dotsc,X^{(t)}$, with $X^{(i)} \in \{0,1\}^n$, and the corresponding outcomes are denoted by $\Yv = (Y^{(1)},\dotsc,Y^{(t)})$, with $Y^{(i)} \in \{0,1\}$.  A given test is allowed to depend on all of the previous outcomes.  Generalizing \eqref{eq:gt_noiseless_model}, we consider the following widely-adopted symmetric noise model:
\begin{equation}
    Y = \bigg(\bigvee_{j \in S} X_j\bigg) \oplus Z, \label{eq:gt_symm_model}
\end{equation}
where $Z \sim \Bernoulli(\rho)$ for some $\rho \in \big(0,\frac{1}{2}\big)$, and $\oplus$ denotes modulo-2 addition.  While symmetric noise is not always realistic, it is commonly considered in the literature on noisy group testing (e.g., \cite{Cha11,Sca15b}).  The main reason for this restriction is that we use the noisy binary search algorithm of \cite{Ben08} as a subroutine, and that work focuses on symmetric noise.  However, we use their algorithm in a ``black-box'' manner, meaning that if any noisy binary search algorithm is devised for other noise models (e.g., Z-channel models \cite{Sca18b}), our algorithms and analysis can be adapted to utilize them accordingly.

Given the tests and their outcomes, a \emph{decoder} forms an estimate $\Shat$ of $S$.  We consider the exact recovery criterion, in which the error probability is given by 
\begin{equation}
    \pe := \PP[\Shat \ne S], \label{eq:pe}
\end{equation}
where the probability is taken over the randomness of the tests $X^{(1)},\dotsc,X^{(t)}$ (if randomized), and the noisy outcomes $Y^{(1)},\dotsc,Y^{(t)}$.  

In the adaptive setting, the algorithm may choose when to stop, and accordingly, the number of tests used may be random.  In such cases, we denote the random number of tests by $T$, and we are interested in characterizing its average, $\EE[T]$.

\subsection{Related Work} \label{sec:related}

While there exist extensive works on group testing (e.g., see \cite{Du93,Ald19} for surveys), we focus our attention here on adaptive algorithms, since this is the focus of our work.

{\bf Noiseless adaptive group testing.} As mentioned above, one of the most well-known strategies for noiseless adaptive group testing is Hwang's generalized binary splitting algorithm \cite{Hwa72}, which works as follows:
\begin{enumerate}
    \item Arbitrarily partition the $n$ items into $k$ groups of size $\frac{n}{k}$.
    \item For each of the $k$ partitions:
    \begin{itemize}
        \item[(a)] Test the entire partition, and if the outcome is negative, then no further steps are performed for this partition.
        \item[(b)] If the outcome is positive, then test the left half of the partition, use the outcome to identify a defective sub-partition (left half or right half), and recursively continue until a single item remains.
        \item[(c)] Declare the item just found as defective, remove it from the partition, and return to Step (a).
    \end{itemize}
\end{enumerate}
This algorithm succeeds using $\big(k\log_2\frac{n}{k}\big)(1+o(1))$ tests whenever $k = o(n)$.  Subsequent works attained the same threshold using distinct four-stage \cite{Dam12} and two-stage \cite{Sca18,Coj19a} adaptive designs, though the two-stage designs only attain $\pe \to 0$ instead of $\pe = 0$.

{\bf Noisy adaptive group testing.} The most related existing theoretical results on noisy adaptive group testing are outlined as follows:
\begin{itemize}
    \item The {\em capacity bound} states that any adaptive algorithm attaining $\pe \to 0$ must have an average number of tests lower bounded as follows when $k = o(n)$ \cite{Bal13}:\footnote{Here and subsequently, the function $\log(\cdot)$ has base $e$, and all information measures are in units of nats.}
    \begin{equation}
        \EE[T] \ge \bigg( \frac{k \log \frac{n}{k}}{ I(\rho) } \bigg) (1+o(1)), \label{eq:t_capacity}
    \end{equation}
    where $I(\rho) = \log 2 - H_2(\rho)$ (with $H_2(\rho) = \rho\log\frac{1}{\rho} + (1-\rho)\log\frac{1}{1-\rho}$ denoting the binary entropy function) is the capacity of a binary symmetric channel with parameter $\rho$.  Moreover, this bound can be strengthened in the regime $k = \Theta(n)$ by replacing $k \log \frac{n}{k}$ by the larger quantity $\log {n \choose k}$.  The capacity bound is usually stated for the case that the number of tests is fixed (i.e., non-random), but the case of a variable number of tests follows similarly from the fact that $I(\rho)$ is the capacity of the binary symmetric channel even when variable-length coding is allowed \cite{Ver10}.
    \item The best known upper bound for a computationally efficient algorithm is given in \cite{Sca19}, and states that there exists a four-stage algorithm that succeeds with probability approaching one using a number of tests satisfying 
    \begin{equation}
        t \le \bigg( \frac{k \log \frac{n}{k}}{ I(\rho) } + \frac{k \log k}{D_2(\rho \| 1-\rho)}\bigg) (1+o(1)), \label{eq:t_existing}
    \end{equation}
    where $D_2(a\|b) = a \log \frac{a}{b} + (1-a)\log\frac{1-a}{1-b}$ is the binary relative entropy function.  
    A refinement of this result is also given in \cite[Thm.~3]{Sca18}, but it based on a computationally expensive brute force search over $n \choose k$ subsets in the first stage, and the difference between the resulting bound and \eqref{eq:t_existing} is almost imperceptible when plotted (see Figure \ref{fig:rho11}).  Since the refined bound is cumbersome to state, we omit the details here.
    \item An additional converse bound is given in \cite{Sca18}, stating that for any algorithm attaining $\pe \to 0$ both in the cases of $k$ defectives and $k-1$ defectives, the number of tests must satisfy\footnote{While \eqref{eq:t_conv} is stated and proved for fixed $t$ in \cite{Sca18}, it also holds in the more general variable-length setting as stated in \eqref{eq:t_conv}.  This is because the first step of the proof in \cite{Sca18} is to write $\sum_{j \in S} \EE[T_j] \le t$, where $T_j$ is the number of tests containing $j$ and no other defective; in the variable-length setting, the same holds with $\EE[T]$ in place of $t$, and the remainder of the proof is unchanged.}
    \begin{equation}
        \EE[T] \ge \bigg( \frac{k \log k}{ \log\frac{1-\rho}{\rho} } \bigg) (1-o(1)). \label{eq:t_conv}
    \end{equation}
    This exceeds the capacity bound \eqref{eq:t_capacity} in sufficiently dense scaling regimes (namely, $k = \Omega (n^{\theta})$ with $\theta$ sufficiently close to one).  In Figure \ref{fig:rho11}, this bound is represented by the diagonal dashed line, whereas the capacity bound is the horizontal dashed line.
\end{itemize}
In the commonly-considered regime $k = \Theta(n^{\theta})$ with $\theta \in (0,1)$, the upper bound \eqref{eq:t_existing} coincides with \eqref{eq:t_capacity} in the limit $\theta \to 0$, and coincides with \eqref{eq:t_conv} when $\rho \to 0$ and $\theta \to 1$ simultaneously \cite{Sca18}, though a gap still remains for fixed $\rho \in (0,1)$ and $\theta \in (0,1)$. 

\begin{figure*}
    \begin{centering}
        \includegraphics[width=0.45\columnwidth]{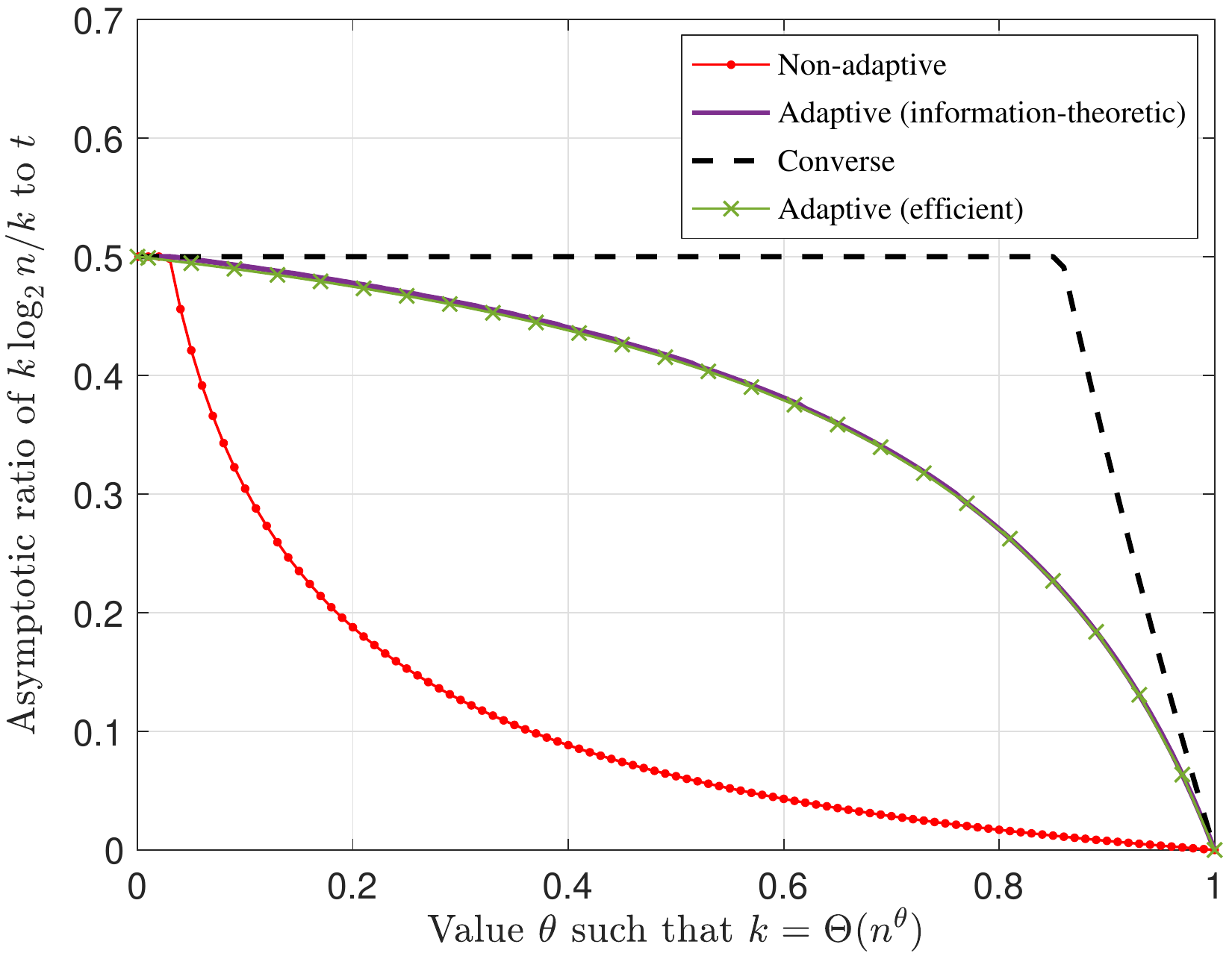}
        \par
    \end{centering}
    
    \caption{Asymptotic performance bounds for noisy group testing with noise level $\rho = 0.11$.  This figure is replicated from \cite{Sca19}, with the non-adaptive achievability result coming from \cite{Sca15b}. \label{fig:rho11}}
\end{figure*}

Beyond the above works, a noisy adaptive version of the GROTESQUE algorithm \cite{Cai13} attains an order-optimal number of tests {\em and decoding time}, but the underlying constants are left unspecified.  On the other hand, there are notable {\em algorithmic} similarities between GROTESQUE and our approach, which we discuss in detail in Section \ref{sec:comparison}.
Recent practical noisy adaptive group testing algorithms include \cite{Cut20,Abr20}, but these do not currently come with theoretical guarantees on the number of tests, which is the main focus of our work.

{\bf Noisy binary search.} A relatively early study of noisy binary search was given by Karp and Kleinberg \cite{Kar07}.  While the constant factors in their theoretical guarantees were not optimized and are too large for our purposes, we will see in Section \ref{sec:comparison} that the algorithm itself (which is based on the method of multiplicative weight updates) is very effective in practice.

Feige \cite{Fei94} devised a noisy binary search algorithm based on traversing a balanced binary tree.  The root node contains all items, each child node contains half of the parent subset, and the leaf nodes are singletons.  The idea is to use the noisy test outcomes to traverse the tree until a defective leaf is found with high probability.  It was shown that for fixed $\rho \in \big(0,\frac{1}{2}\big)$, this algorithm requires $O\left(\log n\right) + O\left(\log \frac{1}{\delta}\right)$ queries, where $n$ is the number of elements in the array being searched, and $\delta$ is the target error probability.

Ben-Or \cite{Ben08} builds on the work of \cite{Fei94} to devise a noisy binary search algorithm with a more precise constant in the $O(\log n)$ term.  For fixed $\rho \in \big(0,1-\frac{1}{\sqrt 2}\big)$, their algorithm solves the noisy binary search problem in $\frac{\log n}{I(\rho)} + O(\log \log n) + O\left(\log \frac{1}{\delta}\right)$ queries.  We use this algorithm as a building block for our group testing algorithms.

{\bf Noisy binary search for group testing.} To our knowledge, the only previous work applying noisy binary search ideas to adaptive group testing is \cite{Bit18}, in the context of private anomaly detection.  However, the analysis therein focuses on the simpler setting of $k=1$, allowing the direct application of feedback communication strategies \cite{Hor63,Bur74}.  In our understanding, this is no longer directly possible in the case that $k > 1$.

In a recent work \cite{Pri21}, a {\em non-adaptive} noisy group testing algorithm based on binary splitting was proposed, building on other recent works handling the noiseless setting \cite{Pri20,Che20}.  As a result of being non-adaptive, the algorithm differs significantly from our adaptive approach.  In addition, the analysis in \cite{Pri21} only establishes scaling laws on the number of tests and decoding time with unspecified constants, whereas in this paper we are interested in the precise constant factors.

\subsection{Contributions}

In this paper, we present three algorithms for noisy adaptive group testing via noisy binary search in increasing order of complexity:
\begin{itemize}
    \item The first variant (Section \ref{sec:appr1}) applies noisy binary search (NBS) in a direct manner, and uses a number of tests whose first term matches the first term in \eqref{eq:t_existing}, but whose second term has an unspecified constant coefficient to $O(k \log k)$.
    \item The second variant (Section \ref{sec:appr2}) combines low-confidence NBS with high-confidence repetition testing, and uses a number of tests nearly matching \eqref{eq:t_existing} (when $\pe \to 0$ sufficiently slowly) but with a higher constant in the second term.
    \item The third variant (Section \ref{sec:appr3}) allows the repetition testing to make some errors and corrects for these in a later step, and uses a number of tests matching \eqref{eq:t_existing} when $\pe \to 0$ sufficiently slowly, while also giving a more general bound specifying the precise dependence on $\pe$.
\end{itemize}
The advantages of our approach over the one in \cite{Sca19} (attaining \eqref{eq:t_existing}) are discussed in more detail in Section \ref{sec:comparison}.  In addition, we provide a numerical example indicating that our general approach can indeed be effective in practice.

From the perspective of our theoretical results, the first two variants mentioned above are primarily introduced as stepping stones towards the third.  On the other hand, the first and simplest variant is seen to perform well experimentally, suggesting that the gaps in the theoretical results may be primarily due the analysis techniques used, rather than inherent strengths and weaknesses in the approaches themselves.  See Section \ref{sec:comparison} for further discussion.

%
%
\section{Preliminaries}

In this section, we introduce a number of useful theoretical and algorithmic tools that will be used throughout the paper.

\subsection{Noisy Binary Search}  \label{sec:nbs}

We first momentarily depart from group testing and discuss the noisy binary search (NBS) algorithm that we will make use of \cite{Ben08}.  The variant of NBS that we consider is as follows (see \cite{Kar07} for more general formulations):  There exists an unknown index $i^* \in \{1,\dotsc,n\}$, and the goal is to locate it via adaptive queries of the form ``Is $i^* \le i$?''.  Each query independently returns the correct answer with probability $1-\rho$, and the incorrect answer with probability $\rho$.

We will make use of the following main result from \cite{Ben08}.

\begin{lem} \label{lem:nbs}
    {\em (NBS Guarantee \cite{Ben08})}
    There exists an NBS algorithm that, given any $\delta \in (0,1)$, succeeds with probability at least $1-\delta$ while satisfying
    \begin{equation}
        \EE[\mathrm{\# queries}] \le \frac{\log n}{ I(\rho) } + O\Big( \log\frac{1}{\delta}\Big) + O(\log \log n).
    \end{equation}
\end{lem}

Note that we re-use the symbol $n$ here as standard notation, but we will typically apply this result to the group testing problem with a smaller quantity, such as $\frac{n}{k}$ in place of $n$ (i.e., considering subsets of the entire set of items).

 The inner workings of the algorithm in \cite{Ben08} are not directly relevant for our main goal of deriving theoretical bounds on the number of tests.  However, in practice, the choice of NBS algorithm is naturally very important.  We are not aware of any works performing a practical implementation of the algorithm in \cite{Ben08}, and doing so may be difficult due to the existence of several ``nuisance'' parameters therein.  In our own experiments (Section \ref{sec:comparison}), we instead use an algorithm from \cite{Kar07} with much looser theoretical guarantees, but excellent practical performance.

\subsection{Modified Noisy Binary Search} \label{sec:modified}

To make NBS more directly useful for group testing, we modify it to consider an array of items in which {\em any} number of items (possibly zero or greater than one) may be defective.  Indexing these items as $\{1,\dotsc,n\}$ (again keeping in mind that later we will substitute $\frac{n}{k}$ or similar), we modify the goal as follows:
\begin{itemize}
    \item If there are one or more defective items in the array, then the algorithm should return the index $i^*$ of the left-most one (e.g., if items $2$, $6$, and $12$ are defective, then $i^* = 2$).
    \item If there are no defective items in the array, then the algorithm should return $\phi$, a symbol indicating that it believes all items to be non-defective.\footnote{In fact, we can prove the same theoretical guarantees for our adaptive group testing algorithms when we omit this modification and assume that the subroutine is only ever run with is at least one defective (treating any other cases as errors).  We still include this modification, since it essentially comes ``for free'' and may be of independent interest.}
\end{itemize}
In addition, the queries are now changed to regular group testing queries, and the overall problem is termed {\em modified noisy binary search} (MNBS).   Defining the notion of error/success probability for an MNBS procedure according to the above criteria, we can obtain the following as a simple consequence of Lemma \ref{lem:nbs}.

\begin{lem} \label{lem:mnbs}
    {\em (MNBS Guarantee)}
    There exists an MNBS algorithm using noisy group tests that, given any $\delta \in (0,1)$, succeeds with probability at least $1-\delta$ while satisfying
    \begin{equation}
        \EE[\mathrm{\# tests}] \le \frac{\log n}{ I(\rho) } + O\Big( \log\frac{1}{\delta}\Big) + O(\log \log n).
    \end{equation}
\end{lem}
\begin{proof}
    The idea is to reduce the problem to regular NBS by (i) always performing group tests with pools of the form $\{1,\dotsc,i\}$, and (ii) adding an $(n+1)$-th ``dummy'' defective to detect when to declare all items as non-defective by returning $\phi$.  The details are given in Appendix \ref{app:mnbs}.
\end{proof}

\subsection{Estimating the Number of Defectives} \label{sec:est_k}

It will be useful to have a procedure for estimating the number of defectives $k$ among items $\{1,\dotsc,n\}$ (and similarly for subsets of items), so that we do not need to assume {\em a priori} knowledge regarding $k$.  While numerous works have given algorithms for estimating $k$ in the noiseless setting \cite{Dam10,Fal16,Bsh18}, analogous results for the noisy setting appear to be very limited.  The following result is based on a simple approach that can likely be improved, but is sufficient for our purposes.

\begin{lem} \label{lem:est_k}
    {\em (Estimation of $k$)}
    Fix any constant $c > 0$.  There exists an adaptive algorithm that, given $n$ items among which $k$ are defective, outputs $\kbar$ satisfying $\frac{\kbar}{2} \le k \le \kbar$ with probability $1 - O(n^{-c})$, using an average of $O(\log k \cdot \log n)$ tests.  Moreover, this can be improved to $(1-\epsilon)\kbar \le k \le \kbar$ for any fixed $\epsilon \in (0,1)$, provided that the implicit constant in the number of tests is suitably modified as a function of $\epsilon$.
\end{lem}
\begin{proof}
    The idea is to iteratively perform a sequence of sufficiently reliable tests to check whether $k \le 2(\sqrt 2)^i$ for $i \in \{0,1,2,\dotsc\}$, until the answer is affirmative.  The details are given in Appendix \ref{app:est_k}.
\end{proof}

\subsection{Repetition Testing} \label{sec:repetition}

Along with noisy binary search itself, we will additionally use suitably-chosen repeated tests to combat the noise.  The relevant auxiliary results for this purpose are given as follows.

{\bf Basic result.} Suppose that we observe a binary value $v$ through a noisy symmetric channel that produces the correct output with probability $1 - \rho$, where $\rho \in \big(0, \frac{1}{2}\big)$.  We would like to determine $v$ with error probability at most $\delta$.  The following lemma gives a standard upper bound on the number of repeated observations that we require.

\begin{lem}
    \label{lem:repetition1}
    {\em (Estimation from Repeated Observations)}
    If we perform $\frac{\log(1/\delta)}{D(\frac{1}{2}||\rho)}$ or more independent noisy observations of $v$ with flip probability $\rho \in \big(0,\frac{1}{2}\big)$ and report the majority outcome, then the error probability is at most $\delta$.  
\end{lem}
\begin{proof}
    See Appendix \ref{app:repetition}.
\end{proof}

{\bf Guarantee for multiple uses.} A standard approach to applying Lemma \ref{lem:repetition1} {\em multiple times} is to choose $\delta$ small enough for a union bound to keep the probability of {\em any} error small.  
%
In one version of our algorithm, we will instead consider an approach that uses fewer tests at the expense of making a small number of errors (which are later corrected).  Formally, we have the following.

\begin{lem}
    \label{lem:repetition2} 
    {\em (Repetition Testing for Multiple Items)}
    Suppose that we perform $\ttil$ individual tests of $k_0$ non-defectives and $k_1$ defectives (i.e.,  $\ttil (k_0 + k_1)$ tests in total).  For any fixed constant $\zeta \in (\rho,1-\rho)$ and any $\delta_0 \in (0,1)$, $\delta_1  \in (0,1)$ and $\epsilon_1  \in (0,1)$, suppose that
    \begin{equation}
        \ttil \ge \max\bigg\{ \frac{\log \frac{k_0}{\delta_0}}{ D_2(\zeta \| \rho) }, \frac{\log \frac{1}{\epsilon_1 \delta_1}}{ D_2(\zeta \| 1 - \rho) } \bigg\}. \label{eq:t_rep2}
    \end{equation}
    Then, if $k_0 = o(k_1)$, we for sufficiently large $n$ that, with probability at least $1 - \delta_0 - \delta_1$, the $(1-\epsilon_1)k_1$ items that returned positive the highest number of times are all defective.
\end{lem}
\begin{proof}
    See Appendix \ref{app:repetition}.
\end{proof}

%
%
\section{Approach 1: Uniform High-Probability Correctness}  \label{sec:appr1}

As a simple starting point, we consider following the structure of Hwang's noiseless algorithm (Section \ref{sec:related}), while applying modified NBS (Section \ref{sec:modified}) with a small enough error probability such that every invocation simultaneously succeeds with high probability.

\subsection{Description of the Algorithm} \label{sec:a1_algo}

The algorithm is described in Algorithm \ref{alg:approach1}.  Here and subsequently, we use the terminology that a partition is {\em empty} if it contains no defectives.

\begin{algorithm}
    \begin{algorithmic}[1]
        \Require Number of items $n$ and defectives $k$, confidence parameter $\delta$
        \State Split the $n$ items into $k$ partitions of size $n/k$ each;\footnotemark
        \State Use repetition testing (Lemma \ref{lem:repetition1}) with confidence $\frac{\delta}{k}$ to test whether each partition contains a defective or not (i.e., for each partition, repeatedly test all the items in the partition together).
        \State For each partition that was declared to be non-empty, run MNBS (Lemma \ref{lem:mnbs}) with confidence parameter $\frac{\delta}{k}$, and add the result (if not given by $\phi$) to the estimated defective set.  Discard all partitions declared empty and all items declared defective, and return to Step 2 (or terminate once all partitions are discarded).
    \end{algorithmic}
    \caption{Description of Approach 1 \label{alg:approach1}}
\end{algorithm}

%

\footnotetext{Since we focus on the regime $k = o(n)$, the effect of rounding is asymptotically negligible, and is ignored in the analysis.  We also note that for the purposes of our theoretical analysis, this partitioning can be arbitrary.}

\subsection{Statement of Theoretical Guarantee} \label{sec:a1_statement}

We state the following recovery guarantee for the above algorithm.

\begin{thm} \label{thm:main_a1}
    Suppose that $k = o(n)$ as $n \to \infty$.  For any $\delta \in (0,1)$ such that $\frac{\delta}{k} = o(1)$, Algorithm \ref{alg:approach1} succeeds with probability at least $1-3\delta$ using an average number of tests satisfying
    \begin{equation}
        \EE[T] \le \bigg(\frac{k \log \frac{n}{k}}{ I(\rho) }\bigg)(1+o(1)) + O\Big( k\log\frac{k}{\delta}\Big). \label{eq:t_a1}
    \end{equation}
\end{thm}
 
This result already matches the capacity bound \eqref{eq:t_capacity} when $\delta \ge \frac{1}{\poly(k)}$ and $k = o(n^c)$ for arbitrarily small $c > 0$ (e.g., $k = \poly(\log n)$), in which case \eqref{eq:t_appr1} reduces to $\big(\frac{k \log \frac{n}{k}}{ I(\rho) }\big)(1+o(1))$.  However, it leaves an unspecified coefficient to the $O(k \log n)$ scaling when $k = \Theta(n^{\theta})$ with $\theta \in (0,1)$, due to the $O\big( k\log\frac{k}{\delta}\big)$ term.  To obtain an explicit constant, we need refined techniques, and these are explored in the subsequent sections.

\subsection{Proof of Theorem \ref{thm:main_a1}} \label{sec:a1_analysis}

{\bf Analysis of correctness.} Consider the algorithm shown above.  We observe that as long as Steps 2 and 3 never make incorrect decisions, the returned output will be correct.  Moreover, as long as these decisions remain correct, Step 2 will be executed $2k$ times (once per defective and once more per partition for when it becomes empty), and Step 3 will be executed $k$ times (once per defective).  Hence, by a union bound over these $3k$ calls and the choice of confidence parameter $\frac{\delta}{k}$, we find that the algorithm succeeds with probability at least $1-3\delta$.

{\bf Number of tests.} The number of tests contributed by the first $2k$ calls to Step 2 is $\frac{2k \log\frac{k}{\delta}}{ D(\frac{1}{2}\|\rho) }$ by Lemma \ref{lem:repetition1}, and the average number of tests contributed by the first $k$ calls to Step 3 is
\begin{equation}
    \frac{k \log \frac{n}{k}}{ I(\rho) } + O\Big( k\log\frac{k}{\delta}\Big) + O(k \log \log n) \label{eq:t_appr1}
\end{equation}
by Lemma \ref{lem:mnbs}.  While further tests may occur if some of these calls to Step 2 and 3 are erroneous, we claim that this only increases the overall average number of tests by a multiplicative $1+O\big(\frac{\delta}{k}\big)$ factor (which behaves as $1+o(1)$ by assumption).  To see this, note that subsequent calls to these steps fail independently of one another, so the number of failures before the first success follows a geometric distribution with success probability $1 - O\big(\frac{\delta}{k}\big)$.  Since the mean of the ${\rm Geometric}(p)$ distribution is $\frac{1}{p}$, the desired claim follows.

Hence, the overall average number of tests is given as in \eqref{eq:t_appr1} up to a multiplicative factor of $1+o(1)$.  Since the first two terms collectively scale as $O(k \log n)$, the final $O(k \log \log n)$ term can also be factored into the multiplication by $1+o(1)$, thus establishing \eqref{eq:t_a1}.

\subsection{Extension to Unknown $k$} \label{sec:a1_ext}

While we presented this approach assuming knowledge of $k$, the analysis goes through unchanged when only an upper bound $\kbar$ is known to the algorithm (again assuming $\kbar = o(n)$), and accordingly $k$ is replaced by $\kbar$ in the number of tests and the error probability.  In fact, the fraction in the first term in \eqref{eq:t_a1} can be $\frac{k \log \frac{n}{\kbar}}{ I(\rho) }$ (with a leading term of $k$ instead of $\kbar$), but the second term $O\big( \kbar \log\frac{\kbar}{\delta}\big)$ may dominate if the upper bound $\kbar$ is much larger than $k$.
  
%
%
\section{Approach 2: Certification of Defectives}  \label{sec:appr2}

A key weakness in Approach 1 is that applying noisy binary search with a confidence of $O\big( \frac{1}{k} \big)$ leads to a number of tests with the $O(k \log k)$ term having an unspecified implied constant, since the main result on NBS in \cite{Ben08} does not specify the coefficient to $O\big(\log\frac{1}{\delta}\big)$.  In this section, we consider a refined approach that applies MNBS with lower confidence (namely, $O\big( \frac{1}{\log n} \big)$) and then uses repetition testing to certify any items that MNBS declares to be defective.  To help better understand this approach, we first outline a simpler version that uses such repetition testing for certifying both defective items and non-defective partitions.

\subsection{Outline of Simpler Approach} \label{sec:simpler} 

A simple refinement of Approach 1 is to first use a higher target error probability (e.g., $O\big(\frac{1}{\log n}\big)$ instead of $O\big(\frac{\delta}{k}\big)$) in noisy binary search and repetition testing, and then only using the smaller $O\big(\frac{\delta}{k}\big)$ confidence parameter for certification, i.e., double-checking that the partition is empty, or double-checking that an estimated defective is indeed defective.

%

For brevity, we do not study this variant in detail, but we mention that it leads to the following number of tests when $\delta$ decays to zero sufficiently slowly:
\begin{equation*}
    \EE[T] \le \bigg( \frac{k \log \frac{n}{k}}{ I(\rho) } + \frac{2k \log k}{ D(\frac{1}{2}\|\rho) }\bigg) (1+o(1)). \label{eq:ap3_bound}
\end{equation*}
While this is a significant improvement on Theorem \ref{thm:main_a1}, the factor of $2$ in the second term is not ideal.  Roughly speaking, this comes from paying a price of $\frac{k \log k}{ D(\frac{1}{2}\|\rho) }$ for certifying each defective (see Lemma \ref{lem:repetition1}), {\em and} a price of $\frac{k \log k}{ D(\frac{1}{2}\|\rho) }$ for certifying each partition as having no remaining defectives.  This motivates the main algorithm of this section described in the following, which only performs high-probability certification of defectives, and reduces the above-mentioned constant from $2$ to $1$.

\subsection{Description of the Algorithm} \label{sec:a2_algo}

We describe the algorithm in two parts.  Note that here we do not assume any prior knowledge of $k$ -- not even an upper bound.

{\bf Description of outer loop.} We first introduce an outer loop, where the goal of each iteration is to identify a constant fraction of the remaining defectives.  The outer loop is described in Algorithm \ref{alg:approach2a}.

\setcounter{algorithm}{0}
\renewcommand\thealgorithm{2\alph{algorithm}} 

\renewcommand\alglinenumber[1]{\footnotesize\Roman{ALG@line}:}

\begin{algorithm}
    \begin{algorithmic}[1]
        \Require Number of items $n$, confidence parameters $\delta_{\rm est}$ and $\delta$,\footnotemark constant $C$
        \State Run the sub-routine for estimating $k$ (Lemma \ref{lem:est_k}), with confidence parameter $\delta_{\rm est}$.  Let the returned value be $\kbar$, which should ideally satisfy $\frac{\kbar}{2} \le k \le \kbar$.
        \State If $\kbar \le C \log n$, then estimate the remaining defectives using Approach 1, with $\kbar$ in place of $k$ and confidence level $\delta$, and terminate the algorithm returning all estimated defectives.
        \State Otherwise, run the group testing subroutine (inner algorithm) below, append the returned subset to the overall set of estimated defectives, remove that subset from further consideration, and return to Step I.
    \end{algorithmic}
    \caption{Outer algorithm for Approach 2 \label{alg:approach2a}}
\end{algorithm}

%

\footnotetext{We will choose $\delta_{\rm est} = O(n^{-c})$ for arbitrarily large $c > 0$, and leave $\delta$ as a free parameter, leading to the algorithm succeeding with probability $1-O(\delta)-O(n^{-c})$.}

{\bf Description of inner algorithm.} In the following, we refer to running repetition testing according to Lemma \ref{lem:repetition1}, and declaring the defectivity status according to a majority vote, as {\em repetition-based certification}. The group testing subroutine for Step III is described in Algorithm \ref{alg:approach2b}.

%

Note that in contrast to Approach 1, we only attempt to find a single defective in each partition, rather than returning to Step 2 in order to seek any further ones; the same effect is instead captured by the outer loop.  In addition, in view of the discussion in Section \ref{sec:simpler}, any partitions believed to be empty are simply ignored, rather than seeking to certify that they are non-defective.

\renewcommand\alglinenumber[1]{\footnotesize\arabic{ALG@line}:}

\begin{algorithm}
    \begin{algorithmic}[1]
        \Require Number of items $n$, estimated upper bound $\kbar$, confidence parameter $\delta$
        \State Randomly split the items into $\kbar$ partitions of size $\frac{n}{\kbar}$, uniformly at random.
        \State Use repetition testing with confidence $\frac{1}{\log n}$ to test whether each partition contains a defective.
        \State For each partition tested in Step 2, if it was declared non-empty, then:
        \begin{itemize}[leftmargin=4ex]
            \item[(a)] Run modified noisy binary search (MNBS) on the partition with confidence parameter $\frac{1}{\log n}$.
            \item[(b)] If MNBS does not return $\phi$,\footnotemark perform repetition-based certification on the returned item with confidence parameter $\frac{\delta}{k}$.  If the item is still declared to be defective, add it to the estimated defective subset.
        \end{itemize} 
        \State Return the estimated subset of defectives.
    \end{algorithmic}
    \caption{Inner algorithm for Approach 2 \label{alg:approach2b}}
\end{algorithm}

\subsection{Statement of Theoretical Guarantee} \label{sec:a2_statement}

We state the following recovery guarantee for the above algorithm.  We focus on the scaling regime $k = \omega(\log n)$, noting that for any smaller $k$, the same result (with the number of tests simplifying to $\frac{k \log \frac{n}{k}}{ I(\rho) } (1+o(1))$) can be obtained via Approach 1, at least under the assumptions on $\delta$ used here.

\footnotetext{Recall that MNBS returns $\phi$ to indicate that it believes all items to be non-defective.}

\begin{thm} \label{thm:main_a2}
    Suppose that $k = \omega(\log n)$ and $k = o(n)$ as $n \to \infty$.  For any $\delta \in (0,1)$ satisfying $\delta \ge e^{-\psi_k}$ for some $\psi_k = o(k)$, and any constant $c > 0$, Algorithm \ref{alg:approach2a} (with suitably-chosen parameters) succeeds with probability $1-O(\delta)-O(n^{-c})$ using an average number of tests satisfying
    \begin{equation}
        \EE[T] \le \bigg( \frac{k \log \frac{n}{k}}{ I(\rho) } + \frac{k \log \frac{k}{\delta}}{D(\frac{1}{2} \| \rho)}\bigg) (1+o(1)). \label{eq:t_appr2}
    \end{equation}
\end{thm}

Compared to Theorem \ref{thm:main_a1}, this result is much more reminiscent of the best known existing result \eqref{eq:t_existing}, though still falls slightly short due to the denominator of $D(\frac{1}{2} \| \rho)$ instead of $D(1-\rho \| \rho)$ in the second term.

We note that the assumption $\delta \ge e^{-\psi_k}$ with $\psi_k = o(k)$ is fairly mild, allowing the target error probability to be nearly exponentially small in $k$.  The additive term $O(n^{-c})$ in the error probability amounts to a slightly stronger restriction, but it is still fairly mild since we allow $c$ to be arbitrarily large.

\subsection{Proof of Theorem \ref{thm:main_a2}} \label{sec:a2_analysis}

{\bf Analysis of correctness.} We first establish the correctness of the algorithm; in accordance with the theorem statement, we would specifically like to show the following:
\begin{equation}
    \text{The algorithm succeeds with probability $1 - O(\delta) - O(n^{-c})$}, \label{eq:target}
\end{equation}
where $c > 0$ may be made arbitrarily large by suitable choices of the algorithm's parameters.

Observe that if the estimate $\kbar$ in the outer algorithm (Step I) remains valid (i.e., with $k'$ current defectives, it holds that $\frac{\kbar}{2} \le k' \le \kbar$), and the inner algorithm always finds at least a constant fraction of the remaining defectives, then the number of outer iterations will be $O(\log k)$.  We will show that this is indeed the case, with high probability.

Regarding $\kbar$, we know that a single invocation of Lemma \ref{lem:est_k} gives the desired event with probability $1-O(n^{-c'})$ for arbitrary $c' > 0$, and hence, its first $O(\log k)$ (or even ${\rm poly}(n)$) invocations simultaneously succeed with probability at least $1-O(n^{-c})$, where again $c > 0$ can be arbitrarily large.  Hence, any errors here only contribute to the $O(n^{-c})$ part in \eqref{eq:target} .

To characterize the inner algorithm, we first show that with high probability, a constant fraction of the partitions are non-empty.  Recall that the partitions are formed uniformly at random.  Fix $\alpha \in (0,1)$, and note that if $k' \in \big[\frac{\kbar}{2}, \kbar\big]$ is the current remaining number of defectives,\footnote{The current number of non-discarded items could be denoted by a different symbol $n' \le n$, but we use $n$, as doing so does not impact the final result.  Moreover, one could always pad the current list of items with dummy non-defectives to keep the total at $n$ throughout.} the probability of $(1-\alpha) \kbar$ specific partitions of size $\frac{n}{\kbar}$ all being empty is
\begin{equation}
    \frac{ {\alpha n \choose k'} }{ {n \choose k'} } \le \frac{ \big( \frac{ne\alpha}{k'} \big)^{k'} }{  \big( \frac{n}{k'} \big)^{k'} } = (e\alpha)^{k'}.
\end{equation}
Setting $\alpha = \frac{1}{9e}$, this simplifies to $9^{-k'} \le 9^{-\kbar/2} = 3^{-\kbar}$.  Taking the union bound over all ${\kbar \choose (1-\alpha) \kbar} \le 2^{\kbar}$ choices of the $(1-\alpha) \kbar$ partitions, it follows that the probability of having below a $\frac{1}{9e}$ fraction of non-empty partitions behaves as $e^{-\Omega(\kbar)}$.  Moreover, since we designed the outer algorithm to ensure that $\kbar \ge C \log n$ for $C$ that we can choose to our liking, the preceding probability reduces to $O(n^{-c'})$ for arbitrary $c' > 0$.  Similarly to the analysis of $\kbar$ above, we can then apply the union bound over the first $O( \log k )$ invocations of the inner algorithm, and the resulting expression only contributes to the $O(n^{-c})$ part in \eqref{eq:target}.

Now, given that at least a $\frac{1}{9e}$ fraction of the partitions are non-empty, we claim that at least half of these non-empty partitions will reach Step 3(b) (i.e., repetition testing will identify that a defective is present in Step 2, and MNBS will not return $\phi$ in Step 3(a)).  This is because when a defective is present, the probability of one or both of Step 2 or Step 3(a) failing is at most $\frac{2}{\log n}$.  Hence, the number of successes dominates a binomial distribution with $\frac{\kbar}{9e}$ trials and success probability $1 - \frac{2}{\log n}$.  By a standard Chernoff bound, this implies that at least $\frac{\kbar}{10e}$ trials succeed, with probability $1 - e^{-\Omega(\kbar)}$ (see Appendix \ref{app:mult_chernoff}).  Again using $\kbar \ge C \log n$ with arbitrarily large $C$, the $e^{-\Omega(\kbar)}$ term here only contributes to the $O(n^{-c})$ term in \eqref{eq:target}.

For each empty partition, Step 3 is only entered with probability $\frac{1}{\log n}$, i.e., the confidence parameter used in Step 2.  In any given invocation of the inner algorithm, there are only $\kbar$ partitions, and from this, we can again apply the Chernoff bound, this time leading to the statement that Step 3(b) is reached for at most $O\big(\frac{\kbar}{\log \log n}\big) = o(\kbar)$ non-defective partitions, with probability $1 - e^{-\Omega(\kbar)}$  (see Appendix \ref{app:mult_chernoff}).  This probability again factors into the $O(n^{-c})$ term.

Next, we sum over the multiple invocations of the inner algorithm, denoting the $\kbar$ value in the $i$-th invocation by $\kbar_i$, and the true number of remaining defectives by $k_i$.  Under the above high-probability events, we have $k_{i+1} \le k_i - \frac{\kbar_i}{10e}$ and $\frac{\kbar_i}{2} \le k_i \le \kbar_i$, and combining these facts gives $\sum_{i} \kbar_i = O(k)$ (noting that $\sum_{i=1}^{\infty} a^i < \infty$ for any $a \in (0,1)$).  Thus, even when we sum all $\kbar_i$ values across all invocations, the total remains linear in $k$.

Now, consider the event of Step 3(b) succeeding on all of its first $O(k)$ invocations.  By the union bound, this occurs with probability at least $1-O(\delta)$, thus leading to the $O(\delta)$ term in \eqref{eq:target}.  In addition, we observe that this implies no errors being made by Step III of the outer algorithm.  This is because whenever a defective is tested in Step 3(b) it is correctly added to the estimate, whereas when a non-defective is tested, it is not added.

Applying the union bound over all of the failure events above, we deduce that the success probability takes the desired form $1-O(\delta) - O(n^{-c})$. The reversion to Approach 1 in Step II additionally contributes to the $O(\delta)$ term.

{\bf Number of tests.} We first study the number of tests conditioned on the high-probability events used in the analysis above, and then turn to the unconditional average.  Under the above high-probability events, we have the following:
\begin{itemize}
    \item In Step I, we use an average of $O(\log k \cdot \log n)$ tests per invocation for estimating the number of defectives (see Lemma \ref{lem:est_k}), for a total of $O(\log^2 k \cdot \log n)$ in the $O(\log k)$ invocations.
    \item Step 2 is invoked once per partition, and having established that $\sum_{i} \kbar_i = O(k)$, the resulting total number of tests is
    \begin{equation}
        \frac{(\sum_i \kbar_i) \log \log n}{ D(\frac{1}{2}\|\rho) } = O(k \log \log n),
    \end{equation}
    where the $\log \log n$ term comes from $\log\frac{1}{\delta'}$ with $\delta' = \frac{1}{\log n}$.
    \item There are $k$ invocations of Step 3 that lead to identifying a defective in Step 3(b), and by Lemma \ref{lem:mnbs}, these incur a total average number of tests given by
    \begin{equation}
        \bigg( \underbrace{\frac{k \log \frac{n}{k}}{ I(\rho) } + O(k \log \log n)}_{\text{Step 3(a)}} + \underbrace{\frac{k \log\frac{k}{\delta}}{ D(\frac{1}{2}\|\rho) }}_{\text{Step 3(b)}} \bigg) (1+o(1)). \label{eq:a3_tests3b}
    \end{equation}
    Here the $1+o(1)$ term is included partly for later convenience, but also due to the subtle issue that in Step 3(a), average the number of tests {\em conditioned on the decision being correct} could in principle be different from the unconditional average.  However, denoting the random number of tests by $\Ttil$ and the error event by $\Ec$, we have $\EE[\Ttil] = \EE[\Ttil | \Ec] \PP[\Ec] + \EE[\Ttil | \Ec^c] \PP[\Ec^c] \ge \EE[\Ttil | \Ec^c] \PP[\Ec^c]$, so the fact that $\PP[\Ec^c] = 1-o(1)$ readily implies $\EE[\Ttil | \Ec^c] \le \EE[\Ttil](1+o(1))$.
    \item Regarding the ``incorrect'' invocations of Step 3 in which the item under consideration is non-defective, we have established that their number is smaller by a factor of $\frac{1}{\log \log n} = o(1)$.  Hence, combining these invocations with those considered in the previous dot point, the number of tests is still \eqref{eq:a3_tests3b} with the differences only amounting to the $1 + o(1)$ term.  By a similar argument, the number of tests corresponding to Step 3(a) failing for a defective item also only contributes to this multiplicative $1 + o(1)$ term. 
    \item In Step II of the outer algorithm, we run Approach 1 with $\kbar = O(\log n)$.  By Theorem \ref{thm:main_a1} (with $\kbar$ in place of $k$ as discussed in Section \ref{sec:a1_ext}), this requires $O((\log n)^2 + \log n \cdot \log \frac{1}{\delta})$ tests.
\end{itemize}
We now sum the number of tests listed above.  Note that \eqref{eq:a3_tests3b} already contributes $O\big(k\log\frac{n}{k} + k\log k\big) = O(k\log n)$.  This means that the $O(k \log \log n)$ and $O(\log^2 k \cdot \log n)$ terms immediately amount to at most multiplication by $1+o(1)$.  Moreover, our assumption on $\delta$ ensures that the same is true of $\log n \cdot \log \frac{1}{\delta}$, since $\delta \ge e^{-o(k)}$ is equivalent to $\log\frac{1}{\delta} \le o(k)$.  Thus, we are left with the simplified number of tests given in \eqref{eq:t_appr2}, as desired.

It remains to argue that the {\em unconditional} average number of tests also satisfies the same bound.  Similarly to the analysis of Approach 1, this follows from the fact that subsequent invocations of all steps use independent tests.  Thus, for any given step, the number of failures until the first success follows a geometric distribution with success probability $1-o(1)$, amounting to an average of $1+o(1)$ invocations until the first success.  The events of these steps failing have no ``knock-on'' effects for future steps, with the possible exception that the incorrect estimation of $k$ could lead to using an unusually large number of tests.  However, the estimation subroutine only fails with probability $O(n^{-c})$ for arbitrarily large $c$, so even the extreme case of returning $\kbar = n$ (and therefore performing one-by-one testing of items, each being in their own ``partition'') only amounts to a negligible $o(1)$ contribution to the overall average number of tests.

%

%
%
\section{Approach 3: Approximate Certification}  \label{sec:appr3}

We are now in a position to state the version of our algorithm that provides the strongest recovery guarantee, matching that of \cite{Sca19} but with notable advantages that we will discuss in Section \ref{sec:comparison}.  Inspired by the adaptive algorithms in \cite{Sca18,Sca19} with 3 or 4 stages, the idea is to use fewer tests in the certification (repetition testing) step at the expense of misclassifying a small number of items, but correcting those mistakes using a small number of additional tests at the end.  We again avoid assuming any prior knowledge on $k$ (even an upper bound).

\subsection{Description of the Algorithm} \label{sec:a3_algo}

In this final variant, we modify Approach 2 to produce a refined variant, detailed in Algorithm \ref{alg:approach3a}.

\setcounter{algorithm}{0}
\renewcommand\thealgorithm{3\alph{algorithm}} 
\renewcommand\alglinenumber[1]{\footnotesize\Roman{ALG@line}:}

\begin{algorithm}
    \begin{algorithmic}[1]
        \Require Number of items $n$, confidence parameters $\delta$ and $\delta_{\rm est}$, individual testing parameter $t_{\rm indiv}$, and constants $C$ and $\epsilon$
        \State Run the sub-routine for estimating $k$ (second part of Lemma \ref{lem:est_k}), with confidence parameter $\delta_{\rm est}$ and approximation parameter $\epsilon$ such that the returned value $\kbar$ should satisfy $(1-\epsilon)\kbar \le k \le \kbar$.\footnotemark  On the first invocation of this step, additionally set $\kbar_{\rm init} = \kbar$ for later use.
        \State If $\kbar \le C \log n$, then estimate the remaining defectives using Approach 1, with $\kbar$ in place of $k$ and confidence parameter $\delta$, then proceed to Step IV.
        \State (If $\kbar > C \log n$) Run the group testing subroutine (inner algorithm) below with parameter $t_{\rm indiv}$, and append the returned set of (item, count) pairs to a list $\Lc$.  Remove all such items from further consideration; return to Step I.
        \State For the (item, count) pairs in $\Lc$ with the $(1-2\epsilon)\kbar_{\rm init}$ highest counts, add each corresponding item to the set of estimated defectives.
        \State For the remaining $|\Lc| - (1-2\epsilon)\kbar_{\rm init}$ items, perform individual testing with confidence parameter $\frac{\delta}{k}$, and declare the items with at least half positive tests as defective.
        \State Return the set of all defectives identified in Steps II, IV, and V.
    \end{algorithmic}
    \caption{Outer algorithm for Approach 3 \label{alg:approach3a}}
\end{algorithm}

\footnotetext{In fact, we only need to let $\epsilon$ be small on the first invocation, and subsequent invocations can revert to $\epsilon = \frac{1}{2}$ as used in Approach 2.}

The inner algorithm is almost identical to that of Approach 2, but we provide the full details for convenience; see Algorithm \ref{alg:approach3b}.

\renewcommand\alglinenumber[1]{\footnotesize\arabic{ALG@line}:}

\begin{algorithm}
    \begin{algorithmic}[1]
        \Require Number of items $n$, estimated upper bound $\kbar$, confidence parameter $\delta$, individual testing parameter $t_{\rm indiv}$
        \State Randomly split the items into $\kbar$ partitions of size $\frac{n}{\kbar}$, uniformly at random.
        \State Use repetition testing with confidence $\frac{1}{\log n}$ to test whether each partition contains a defective or not.
        \State For each partition tested in Step 2, if it was declared non-empty, then:
        \begin{itemize}[leftmargin=4ex]
            \item[(a)] Run modified noisy binary search (MNBS) on the partition with confidence parameter $\frac{1}{\log n}$.
            \item[(b)] If MNBS does not return $\phi$, perform $t_{\rm indiv}$ individual tests on the returned item, and append the resulting (item, count) pair (with ``count'' being the number of positive tests) to the output list.
        \end{itemize} 
        \State Return the produced list of (item, count) pairs.
    \end{algorithmic}
    \caption{Inner algorithm for Approach 3 \label{alg:approach3b}}
\end{algorithm}

%

\subsection{Statement of Theoretical Guarantee} \label{sec:a3_statement}

We state the following recovery guarantee for the above algorithm.  Similarly to Theorem \ref{thm:main_a2}, we focus on the scaling regime $k = \omega(\log n)$.

\begin{thm} \label{thm:main_a3}
    Suppose that $k = \omega(\log n)$ and $k = o(n)$ as $n \to \infty$.  For any $\delta \in (0,1)$ satisfying $\delta \ge e^{-\psi_k}$ for some $\psi_k = o(k)$, any positive $\delta_0,\delta_1$ such that $\delta_0 + \delta_1 \le O(\delta)$, and any constant $c > 0$, Algorithm \ref{alg:approach3a} with suitably-chosen $t_{\rm indiv}$ (see \eqref{eq:t_indiv2} below), succeeds with probability $1-O(\delta)-O(n^{-c})$ using an average number of tests satisfying
    \begin{equation}
        \EE[T] \le \inf_{\zeta \in (\rho,1-\rho)} \bigg( \frac{k \log \frac{n}{k}}{ I(\rho) } + \max\bigg\{ \frac{k \log \frac{k}{\delta_0}}{D_2(\zeta \| \rho)}, \frac{k \log \frac{1}{\delta_1}}{D_2(\zeta \| 1-\rho)} \bigg\}\bigg) (1+o(1)). \label{eq:t_appr3}
    \end{equation}
\end{thm}

While the minimization over $\zeta$ and the freedom in choosing $\delta_0$ and $\delta_1$ make this expression somewhat more complicated than our previous results, we can highlight two important special cases:
\begin{itemize}
    \item[(i)] Setting $\delta_0 = \delta_1 = \delta$ and $\zeta = \frac{1}{2}$, we recover the bound in Theorem \ref{thm:main_a2}. 
    \item[(ii)] If $\delta_0 = \delta_1 = \delta$ and $\delta \to 0$ with $\delta = k^{-o(1)}$ (e.g., $\delta = \frac{1}{{\rm poly}(\log k)}$), then the $\max\{\cdot,\cdot\}$ in \eqref{eq:t_appr3} is asymptotically attained by the first term for any fixed $\zeta \in (\rho,1-\rho)$.  Hence, we can take $\zeta$ to be arbitrary close to $\rho$, and the overall bound simplifies to 
        \begin{equation}
        \EE[T] \le \bigg( \frac{k \log \frac{n}{k}}{ I(\rho) } + \frac{k \log k}{D_2(1-\rho \| \rho)} \bigg) (1+o(1)). \label{eq:t_appr3a}
    \end{equation}
    Thus, we recover the existing bound \eqref{eq:t_existing} (note that $D_2(1-\rho \| \rho) = D_2(\rho \| 1-\rho)$).  A more detailed comparison is given in Section \ref{sec:comparison}.
\end{itemize}

\subsection{Proof of Theorem \ref{thm:main_a3}} \label{sec:a3_analysis}

{\bf Re-used findings.} Since the inner algorithm and Steps I-III of the outer algorithm coincide with those of Approach 2, we are able to re-use the following high-probability findings from the proof of Theorem \ref{thm:main_a2}:
\begin{itemize}
    \item The estimate $\kbar_i$ in each iteration $i$ satisfies the desired bound $(1-\epsilon)\kbar_i \le k_i \le \kbar_i$, where $k_i$ is the true number of remaining defectives.
    \item Step 3(b) is reached for  a constant fraction of defectives in each invocation of the inner algorithm, and hence $k_{i+1} \le (1-\Omega(1)) k_i$, and there are $O(\log k)$ total invocations.
    \item Step 3(b) is reached by all defectives (except those deferred to Approach 1 in Step II), but is only reached by $o(k)$ non-defectives.
    \item Assuming $t_{\rm indiv} = \Omega(\log k)$ (which will indeed be the case when we set its value), the overall number of tests resulting from Steps I--III (including those in the inner algorithm) is at most
    \begin{equation}
        \bigg( \frac{k\log\frac{n}{k}}{I(\rho)} + k t_{\rm indiv} \bigg)(1+o(1)), \label{eq:t_reused}
    \end{equation}
    and the overall contribution to the error probability from Steps I-III is $O(\delta) + O(n^{-c})$.
\end{itemize}

{\bf Analysis of correctness.} With the above findings in place, we seek to apply Lemma \ref{lem:repetition2} to deduce that Step IV succeeds, in the sense that the top $(1-2\epsilon) \kbar_{\rm init}$ ranked items in the list $\Lc$ are indeed defective.

By the assumption $k = \omega(\log n)$ and the third finding above, we have that $\Lc$ contains $k_1 = k(1-o(1))$ defectives, and $k_0 = o(k)$ non-defectives; in particular, compared to $k = \omega(\log n)$, the $C \log n$ (or fewer) items deferred to Approach 1 can be absorbed into the $1-o(1)$ term.  Hence, the condition $k_0 = o(k_1)$ in Lemma \ref{lem:repetition2} is satisfied.  In accordance with \eqref{eq:t_rep2}, and crudely upper bounding $k_0 \le k$, we fix $\epsilon_1 \in (0,1)$ and $\zeta \in (\rho,1-\rho)$ and set
\begin{equation}
    t_{\rm indiv} = \max\bigg\{ \frac{\log \frac{k}{\delta_0}}{ D_2(\zeta \| \rho) }, \frac{\log \frac{1}{\epsilon_1 \delta_1}}{ D_2(\zeta \| 1 - \rho) } \bigg\}. \label{eq:t_indiv1}
\end{equation}
Noting that $\epsilon_1$ is a fixed positive constant (but arbitrarily small), we claim that this simplifies to
\begin{equation}
    t_{\rm indiv} = \max\bigg\{ \frac{\log \frac{k}{\delta_0}}{ D_2(\zeta \| \rho) }, \frac{\log \frac{1}{\delta_1}}{ D_2(\zeta \| 1 - \rho) } \bigg\}(1+o(1)). \label{eq:t_indiv2}
\end{equation}
This is immediate when $\delta_1 = o(1)$, and otherwise, the assumption $k \to \infty$ implies that the maximum is achieved by the first term in both \eqref{eq:t_indiv1} and \eqref{eq:t_indiv2} anyway.

Recalling the assumption $\delta_1 + \delta_2 = O(\delta)$, Lemma \ref{lem:repetition2} implies that with probability $1 - O(\delta)$, the top $(1-\epsilon_1)k_1$ ranked items are all defective.  Having established that $k_1 = k(1-o(1))$ and $(1-\epsilon)k \le \kbar_{\rm init} \le k$, we observe that we can choose $\epsilon_1$ sufficiently small such that $(1-2\epsilon) \kbar_{\rm init} \le (1-\epsilon_1)k_1$.  Hence, the top $(1-2\epsilon) \kbar_{\rm init}$ ranked items are all defective, as desired for Step IV of the algorithm to succeed.

For Step V, the remaining number of defectives is $k - (1-2\epsilon) \kbar_{\rm init} \le C' \epsilon k$, where $C'$ is an absolute constant, and the remaining number of non-defectives is $o(k)$.  We apply the basic form of repetition testing in Lemma \ref{lem:repetition1} with confidence $\frac{\delta}{k}$.  By a union bound, every item is classified correctly with probability at least $1-\delta$ when we test each item $\frac{ \log\frac{k}{\delta} }{ D(\frac{1}{2} \| \rho) }$ times, leading to a total number of tests for Step V given by
\begin{equation}
    C' \epsilon \cdot \frac{ k \log\frac{k}{\delta} }{ D(\frac{1}{2} \| \rho) }. \label{eq:t_indiv3}
\end{equation}

{\bf Number of tests.} We combine \eqref{eq:t_reused} with \eqref{eq:t_indiv2} and \eqref{eq:t_indiv3}.  Since $\epsilon$ can be arbitrarily small, the term \eqref{eq:t_indiv3} is arbitrarily small compared to \eqref{eq:t_indiv2}, and can be absorbed into the multiplicative $1+o(1)$ term.   Choosing $\zeta \in (\rho,1-\rho)$ to minimize the total, we obtain the desired bound \eqref{eq:t_appr3}.  We have derived this bound conditioned on high-probability events, but the unconditional average again follows via an analogous argument to that of Approach 2.

%
%
\section{Comparisons to Existing Approaches} \label{sec:comparison}

In this section, we first compare our approach to that of \cite{Sca19} (in which  \eqref{eq:t_existing} was attained), and then provide an experimental example indicating that NBS-based group can be effective in practice.

\subsection{Comparison to Coding-Based Approaches}

{\bf Comparison to \cite{Sca19}.} As we already noted, if our goal is simply to attain $\pe \to 0$ as $n \to \infty$ (with $k \to \infty$ and $k = o(n)$), then the number of tests required by Approach 3 in this paper matches that of \cite{Sca19}.  An immediate advantage of our approach is that we also explicitly characterize how the number of tests depends on the target error probability $\delta$, whereas adapting the analysis of \cite{Sca19} to characterize this dependence appears to be difficult.

In addition, an inspection of the approach and analysis in \cite{Sca19} reveals the following limitations:
\begin{itemize}
    \item[(i)] As a subroutine, the algorithm in \cite{Sca19} uses a capacity-achieving channel code with block length $O(\log n)$.  In view of the fundamental limits imposed at finite block lengths \cite{Pol10}, this short block length is likely to require a significant backoff from capacity in practical problem sizes.  For instance, even $n = 10^{10}$ gives $\log n \approx 23$, which is an extremely short block length for a communication code.\footnote{Without going into detail, we note that the implied constants in the notation $O(\log n)$ are also not large enough to overcome this limitation when aiming for an asymptotically near-optimal number of tests.}
    \item[(ii)] The analysis of the error probability in \cite{Sca19} contains terms of the form $\frac{k^{-\epsilon}}{\alpha}$, where both $\epsilon$ and $\alpha$ are taken to zero at the end of the proof.  This indicates that $k$ may need to be extremely large (and $n$ is in turn even larger) to attain a small error probability.
\end{itemize}
Regarding item (i), our use of the NBS subroutine overcomes this limitation and exploits the benefits of full adaptivity and variable-length stopping, analogously to how variable-length feedback communications codes require significantly less finite-length backoff from capacity than fixed-length codes without feedback \cite{Pol11}.  The caveat is that requiring full adaptivity serves as a notable disadvantage compared to the approach of \cite{Sca19}, which uses only 4 stages of adaptivity.

In fairness, our analysis may have some similar limitations to item (ii) above (e.g., absorbing $O(\log \log n)$ factors into lower-order terms), but nevertheless, we believe it to be a significantly more practical approach overall.  In particular, it does not suffer from any notion of ``short effective code length'' or similar, and it permits the use of NBS algorithms in a black-box manner.  In the following subsection, we corroborate this discussion with an experimental example.

Another minor advantage is that our algorithms (Approaches 2 and 3) do not use any prior knowledge regarding the number of defectives $k$, though it should be straightforward to overcome this limitation by first estimating $k$ and then applying the approach in \cite{Sca19}.

{\bf Comparison to GROTESQUE.} The theoretical guarantees for GROTESQUE \cite{Cai13} are less comparable to ours, since they only seek scaling laws and not precise constants.  Nevertheless, it is interesting to note the algorithmic similarities between our approach (particularly Approach 1 in Section \ref{sec:a1_algo}) and that of GROTESQUE (adaptive version).  Both share a similar high-level structure, with the following main differences:
\begin{itemize}
    \item[(i)] We use repetition testing to distinguish between the cases of $\{0,\ge1\}$ defectives in each partition, whereas GROTESQUE uses i.i.d.~Bernoulli$\big(\frac{1}{2}\big)$ testing to distinguish between the cases of $\{0,1,\ge 2\}$ defectives, and only proceeds with 1-sparse recovery if the answer is ``1''.
    \item[(ii)] We use noisy binary search to identify the left-most defective in the ``$\ge 1$'' case, whereas GROTESQUE uses an expander code of length $O(\log n)$ to identify the unique defective in the ``$1$'' case.
    \item[(iii)] Similarly to our Approach 2, GROTESQUE continually re-shuffles the partitions in order to ensure that each item ends up in the ``1'' case after enough outer iterations.
\end{itemize}
Regarding item (i), our sub-problem appears to be easier to solve, since distinguishing the cases ``0'' vs.~``$\ge 1$'' can be done with a simple majority vote, whereas distinguishing ``1'' vs.~``0 or $\ge 2$'' requires more carefully choosing two appropriate thresholds (deciding ``1'' if the fraction of positive tests is in between them).
More importantly, item (ii) implies that the same practical limitation as that of \cite{Sca19} discussed above also applies to GROTESQUE.  While our approach has the advantage of avoiding this limitation, this again comes with the caveat that one should also consider the number of stages of adaptivity --  our approach is fully adaptive, but GROTESQUE only requires $O(\log n)$ stages.

\subsection{Experimental Example}

Here we present a simple proof-of-concept experiment to show that our general approach can be effective; note that we do not seek to be comprehensive in terms of diverse experimental settings or detailed comparisons. 
We follow the experimental setup of \cite[Sec.~3.7]{Ald19}, and compare against non-adaptive group testing with the two best-performing decoding algorithms therein: belief propagation (BP) \cite{Sed10} and linear programming (LP) \cite{Mal12}.  For the non-adaptive test matrix, we consider both an i.i.d.~design and a constant-column weight design, with the parameters chosen such that each test has a $50\%$ chance of being positive, as suggested by numerous theoretical studies \cite{Mez08,Sca15b,Joh16,Coj19,Coj19a}.  Due to the practical challenges discussed above, we do not attempt to compare to the adaptive algorithms in \cite{Sca19,Cai13}. 

For a practical implementation, we aim to keep our adaptive algorithm as simple as possible, and accordingly use Approach 1 (Section \ref{sec:a1_algo}) with the following modifications:
\begin{itemize}
    \item In Step 2 (repetition testing), we simply set the number of repetitions to a fixed odd number $r$, separately considering the values $r \in \{3,5,7,9,11,13\}$.  To slightly reduce the number of tests, we stop early when the outcome is already determined (e.g., if $r = 5$ and the first three outcomes are the same, then we can skip the final two).
    \item In Step 3 (applying NBS), we apply Karp and Kleinberg's multiplicative weights algorithm \cite{Kar07},\footnote{We used the implementation available at \url{https://github.com/adamcrume/robust-binary-search}.} which we found to be the most practical despite its analysis providing highly suboptimal constant factors.  This algorithm has a single parameter $\delta_{\rm NBS}$, which we set to $\delta_{\rm NBS} = \frac{\delta}{3k}$ with various $\delta \in [0.002,5]$.  We do not make use of the modification in which $\phi$ may be returned, as we found this to make little difference here.
\end{itemize}
Approaches 2 and 3 appear to be more difficult to implement, and since our experiments are only meant to serve as a simple proof of concept, we do not attempt to do so; see Remark \ref{rem:approaches} below for further discussion.

Figure \ref{fig:experiment} plots the experimental results for $n=500$ and $k=10$, with a noise level of $\rho = 0.05$.  The $x$-values and $y$-values for our adaptive algorithm are averaged over $10^5$ trials.  For the non-adaptive algorithms, the $x$-values are fixed, and the $y$-values are averaged over $10^3$ trials (due to their more expensive computation).  For approximate recovery, by ``fraction of mistakes'', we refer to the quantity
\begin{equation}
    \frac{1}{k} \max\big\{ |S \setminus \Shat|, |\Shat \setminus S| \},
\end{equation}
which is zero when $\Shat = S$ and one when none of the defectives are found.

Under both recovery criteria, we see that adaptivity is able to provide significant gains over the state-of-the-art non-adaptive algorithms (at least for suitably-chosen $\delta$ and $r$), despite having only considered the simplest version of our algorithm with no fine-tuning.  We note that the distinct ``curves'' for the adaptive algorithm correspond to the six different values of $r \in \{3,5,7,9,11,13\}$ mentioned above.  We also found that the fraction of mistakes for the adaptive algorithm is often close to $\frac{1}{k}$ times the error probability, indicating that when an error is made, it may only be due to having a single false positive or a single false negative.

\begin{figure}
    \begin{centering}
        \includegraphics[width=0.42\columnwidth]{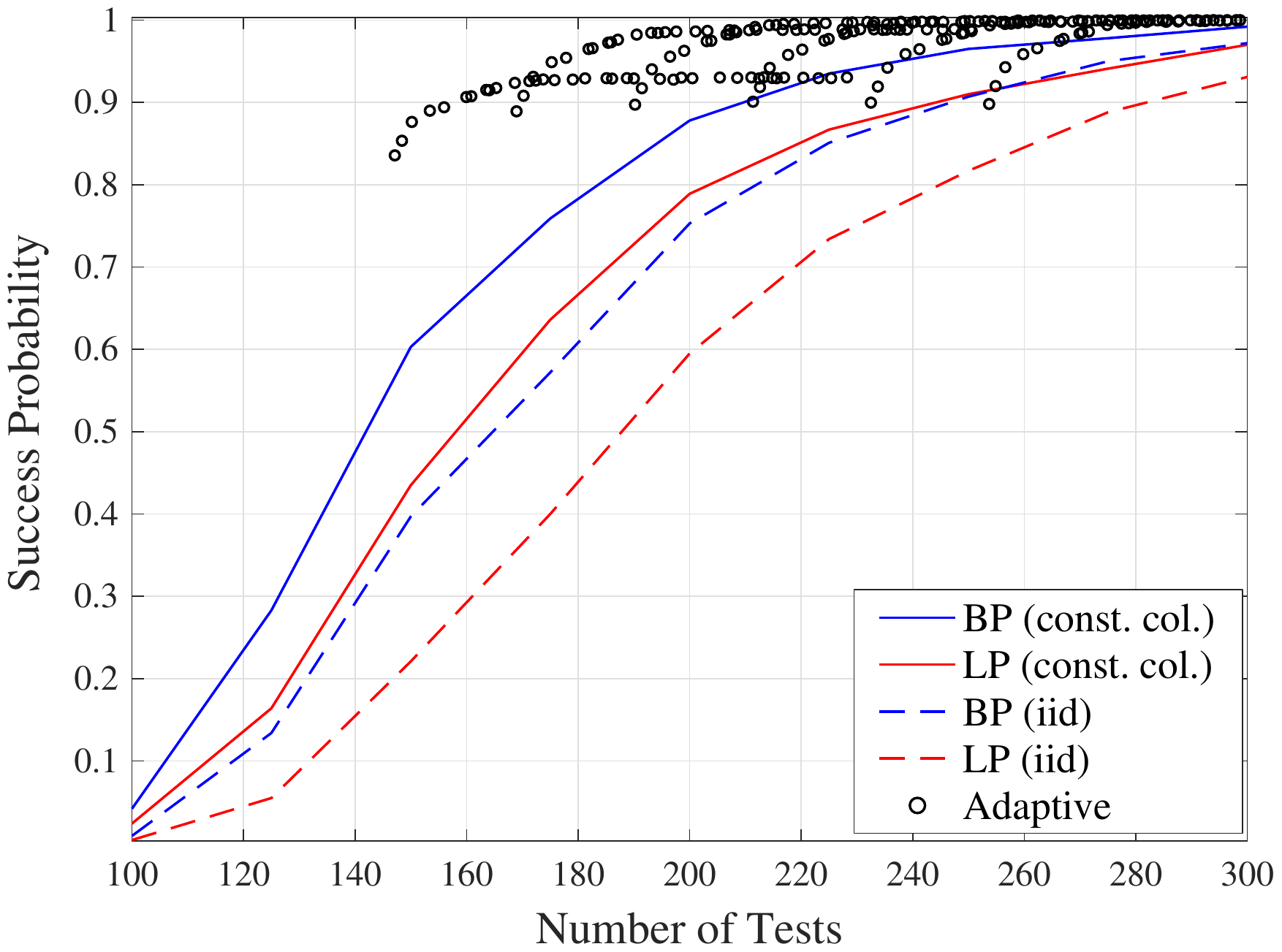} \quad
        \includegraphics[width=0.44\columnwidth]{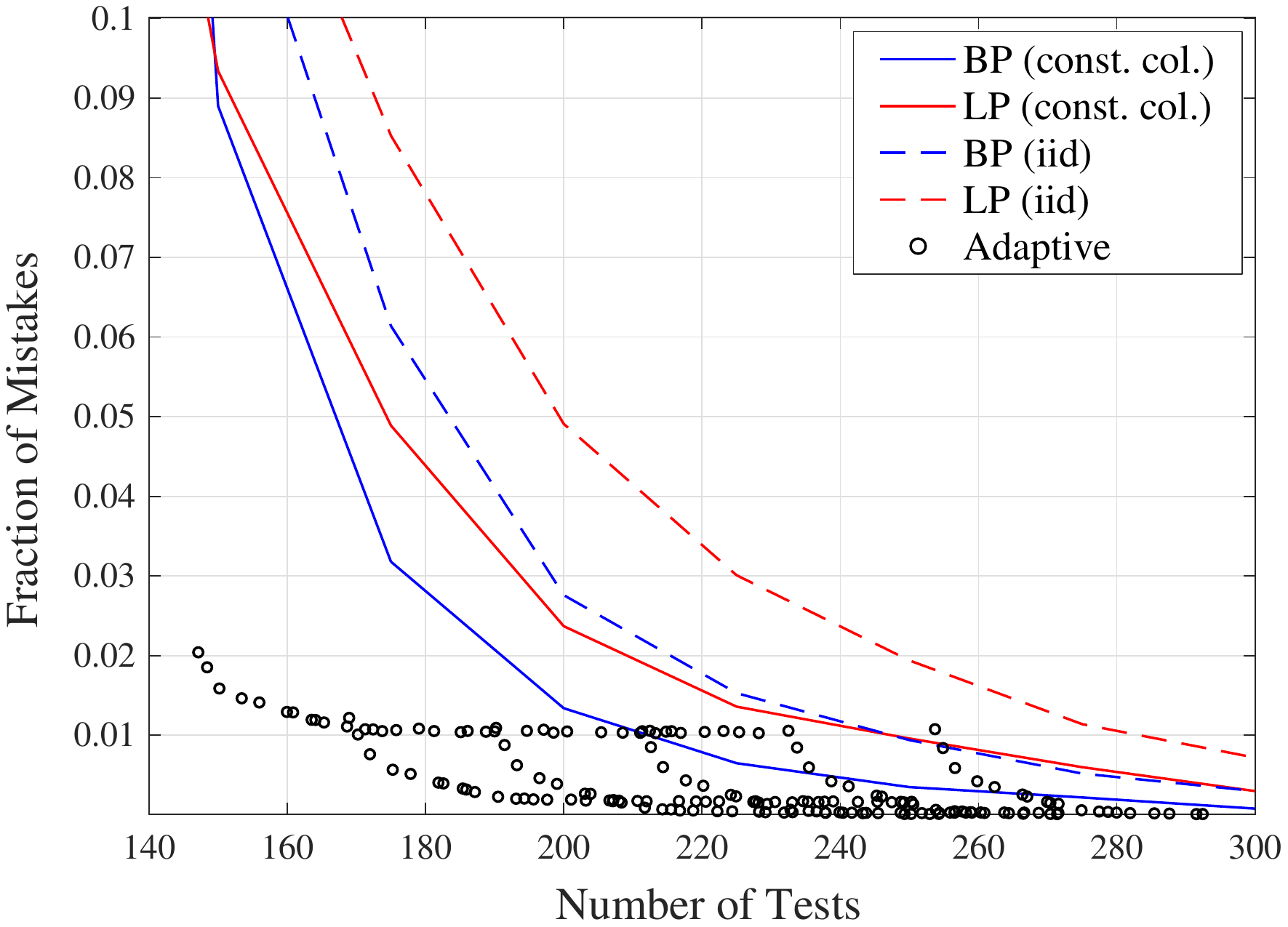}
        \par
    \end{centering}
    
    \caption{Numerical experiments for exact recovery (left) and approximate recovery (right). \label{fig:experiment}}
\end{figure}

A slight caveat to these performance gains is that the gap may be reduced if {\em hard limits} are placed on the number of tests used in each trial.  In such cases, it is crucial to understand not only the average number of tests used by the adaptive algorithm, but also the variance.  To visualize this, we give an example histogram of the number of tests used by one configuration of $(r,\delta)$ among $10^6$ trials (namely, $r=5$ and $\delta=0.2$).  We observe that the standard deviation is small compared to the average number of tests, though non-negligible.

\begin{rem} \label{rem:approaches}
    A notable difficulty in implementing Approaches 2 and 3 is that the theoretical choices of the extra parameters may be less suitable at practical problem sizes (e.g., $\delta_{\rm est}$, $C$, the $\frac{1}{\log n}$ confidence term, and $\epsilon$ in Approach 3), amounting to having significantly more parameters to tune instead of just $r$ and $\delta$.  The main reason for requiring these approaches in our theoretical analysis is to overcome the unspecified constant in the $O\big( \log\frac{1}{\delta} \big)$ dependence in \cite{Ben08}.  However, the algorithm in \cite{Kar07} appears to behave very favorably with respect to decreasing $\delta$ in practice (significantly better than the associated theoretical bounds).  Thus, it may be that the requirement of two confidence levels (low-confidence NBS and high-confidence certification) is purely theoretical and not beneficial for practical problem sizes; answering this definitively is beyond the scope of our work. 
\end{rem}

\begin{figure}
    \begin{centering}
        \includegraphics[width=0.44\columnwidth]{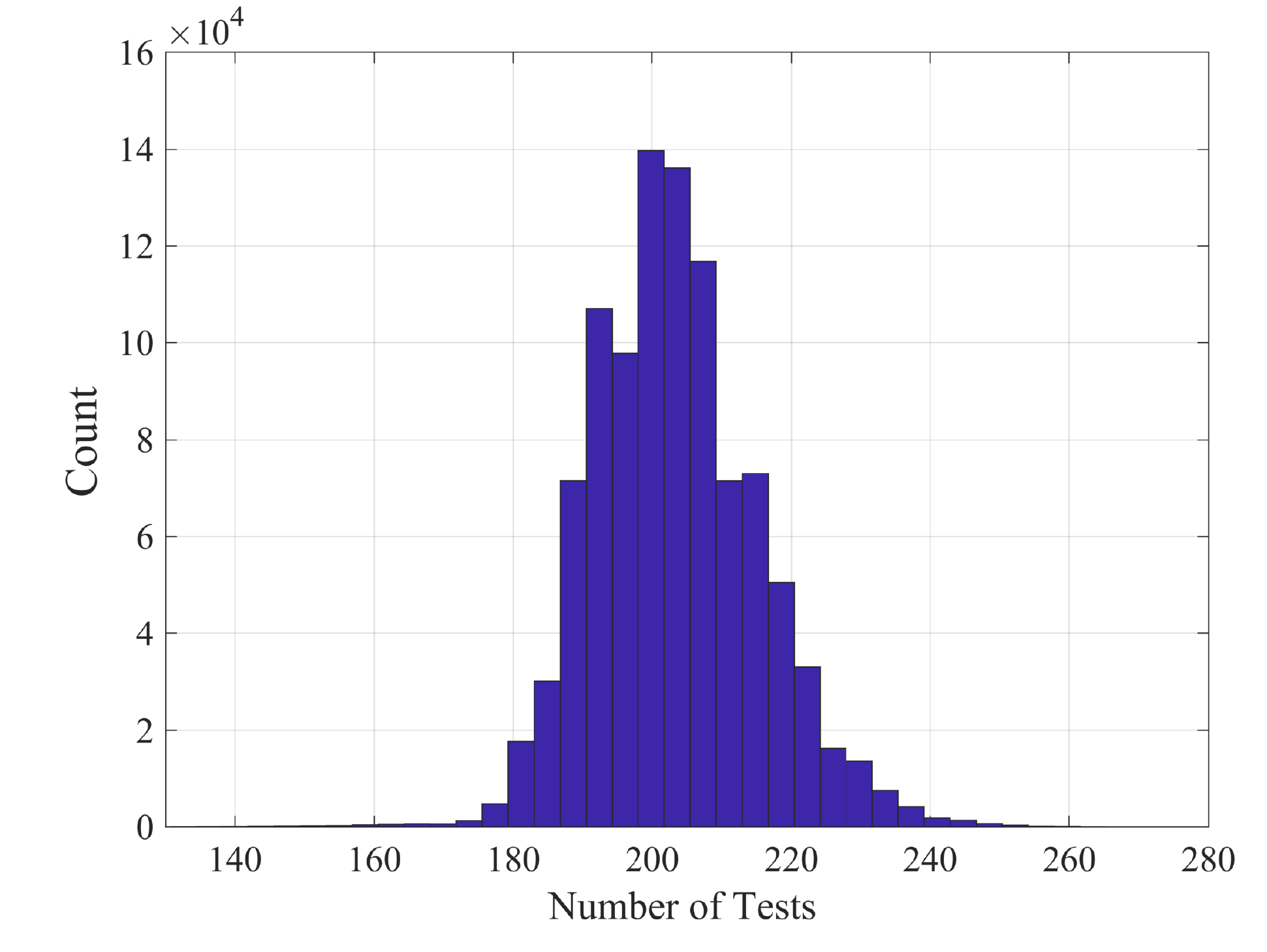}
        \par
    \end{centering}
    
    \caption{Example histogram of the number of tests used by the adaptive algorithm. \label{fig:histogram}}
\end{figure}

%
%
\section{Conclusion}

We have introduced and analyzed noisy group testing algorithms that rely on noisy binary search as a subroutine, serving as natural noisy counterparts to the famous splitting algorithm of Hwang \cite{Hwa72}.  The most sophisticated variant of our algorithm attains the best known theoretical guarantee attained by any practical algorithm, while also overcoming a key practical limitation and explicitly specifying dependence on the target error probability.  We have observed that even the simplest variant of our algorithm can perform very well in practice, and in future work, it may be of interest to refine this to attain further improved practical variants.

%
%
\bigskip
\appendix

\subsection{Proof of Lemma \ref{lem:mnbs} (Modified Noisy Binary Search)} \label{app:mnbs}

We describe the two ways in which the algorithm of \cite{Ben08} is modified as follows.

{\bf Implementing NBS queries via group testing queries.} Noisy binary search (NBS) works with an array of length $n$, and seeks to identify an unknown index $i^* \in \{1,\dotsc,n\}$ via noisy answers to questions of the form ``Is $i^* \le i$?''  In contrast, in group testing, we have a binary vector in $\{0,1\}^n$ (where $1$ indicates defectivity), and we can observe noisy ``OR'' queries on any chosen subset.

To apply NBS using group testing queries, we need to account for the fact that there may be {\em multiple} defectives (or no defectives, but this is handled separately below).  To do so, we specifically seek to {\em identify the left-most defective in the list}, thus making the notion of a specific item $i^*$ well-defined.  Then, if we want to ask the query ``Is $i^* \le i$?'', we simply test all of the items in $\{1,\dotsc,i\}$.  Thus, we can search for $i^*$ using NBS in a black-box manner, even if $i^*$ is not the only defective item.

{\bf Allowing the possibility of no defectives.} In some cases, we will apply the NBS subroutine on arrays where every item is non-defective, and we would like NBS to be able to detect this scenario.  To achieve this, we add a dummy defective to the end of the array, i.e., if the original length was $n$, then we artificially add another defective indexed by $n+1$, and manually set the test outcome to be ${\rm Bernoulli}(1-\rho)$ whenever this item is included.  Then, if NBS returns $n+1$ as its estimate of $i^*$, we take this to indicate that no defectives are present among indices $\{1,\dotsc,n\}$, and return $\phi$.  The bound on the number of tests from Lemma \ref{lem:nbs} is unchanged, because the change from $n$ to $n+1$ only affects the higher-order terms.

\subsection{Proof of Lemma \ref{lem:est_k} (Estimating the Number of Defectives)} \label{app:est_k}

We first analyze a useful subroutine, and then present the full algorithm.

{\bf A useful subroutine.} Let $k_0 \ge 2$ be a putative number of defective items, and consider the goal of declaring whether $k \ge k_0$ or $k \le \frac{k_0}{\sqrt 2}$.  For any $k \in \big( \frac{k_0}{\sqrt 2}, k_0 \big)$, either declaration is considered adequate.  We make use of i.i.d.~Bernoulli testing, in which each item is independently placed into each test with probability $\frac{\nu}{k_0}$ for some $\nu > 0$.

We set $\nu = \nu_0 $, where $\nu_0$ is defined to be the parameter that would lead to i.i.d.~Bernoulli testing having positive and negative tests outcomes being equally likely if exactly $k_0$ defectives were present, i.e., $\big(1 - \frac{\nu_0}{k_0}\big)^{k_0} = \frac{1}{2}$ (e.g., see \cite{Tru20}).  It follows immediately that when $k \ge k_0$, the probability of a positive test is at least $\frac{1}{2}$.

On the other hand, if $k \le \frac{k_0}{\sqrt 2}$, then in the {\em noiseless} group testing model (whose output we denote by $U$ to distinguish it from $Y$), the probability of a positive test satisfies 
\begin{align}
    \PP[U = 1] 
    &\le 1 - \bigg(1 - \frac{\nu_0}{k_0}\bigg)^{k_0/{\sqrt 2}} < \frac{1}{2},
\end{align}
where the strict inequality holds since replacing $k_0/{\sqrt 2}$ by $k_0$ would give exactly $\frac{1}{2}$ due to our choice of $\nu$.  The precise gap to $\frac{1}{2}$ is not important for our purposes. 

Next, we observe that a strict gap to $\frac{1}{2}$ in the noiseless case implies the same in the noisy case, since if $U$ is flipped with probability $\rho$ to produce $Y$, then 
\begin{align}
    \PP[Y = 1] 
    &= \PP[U = 1](1-\rho) +  (1-\PP[U = 1]) \rho \\
    &= \PP[U = 1]\big(1 - 2\rho) + \rho.
\end{align}
The right-hand side is strictly increasing with respect to $ \PP[U = 1]$, and equals $\frac{1}{2}$ when $\PP[U = 1] = \frac{1}{2}$, which establishes the desired claim.

Thus, distinguishing between $k \ge k_0$ and $k \le \frac{k_0}{\sqrt 2}$ simply amounts to distinguishing independent Bernoulli random variables with parameters known to be at least $\frac{1}{2}$ and strictly less than $\frac{1}{2}$, respectively.  By standard Chernoff bounds, this can be achieved with success probability decaying exponentially in the number of tests.  In particular, we can attain success probability $1 - O(n^{-c_0})$ for any constant $c_0 > 0$, using $O(\log n)$ tests with an implied constant depending on $c_0$.

{\bf Algorithm for estimating $k$.} We turn the above procedure into an algorithm that estimates the total number of defectives to within a factor of $\frac{1}{2}$.  To do this, we simply run the above procedure with $k_0 = 2$, then again with $k_0 = 2\sqrt{2}$ if the first step declares $k \ge 2$, and so on until the procedure declares $k \le \frac{k_0}{\sqrt 2}$, at which point we return $\kbar = k_0$.  As long as no errors are made, we have the following:
\begin{itemize}
    \item It holds that $\frac{\kbar}{2} \le k \le \kbar$ (otherwise, we would contradict one of the final two decisions made);
    \item The procedure is called $O(\log k)$ times.
\end{itemize}
By considering a geometric random variable counting the number of possible failures after $k_0$ indeed exceeds $2k$, the average number of calls is also $O(\log k)$, for a total of $O(\log k \cdot \log n)$ tests on average.

In addition, success is guaranteed as long as the first $O(\log k)$ calls to the procedure succeed, so by the union bound and a suitable choice of $c_0$ above, we are guaranteed to succeed with probability $1 - O(n^{-c})$ for any fixed $c > 0$.  This establishes the first claim in Lemma \ref{lem:est_k}.

The second part of the lemma with a general value of $\epsilon \in (0,1)$ follows similarly; the above analysis corresponds to $\epsilon = \frac{1}{2}$, but there is no significant change for any other fixed value in $(0,1)$.

\subsection{Proofs of Lemmas \ref{lem:repetition1} and \ref{lem:repetition2} (Repetition Testing)} \label{app:repetition}

\begin{proof}[Proof of Lemma \ref{lem:repetition1}]
    Suppose that we use a number of repeated observations given by some generic value $\ttil$, and that we estimate that $v$ is $1$ if more than half of the tests returned $1$, and $0$ otherwise.  Then by the Chernoff bound for binomial random variables \cite[Sec.~2.2]{Bou13}, the error probability is at most $e^{-\ttil D_2(\frac{1}{2}||\rho)}$.  We set $\ttil = \frac{\log(1/\delta)}{D_2(\frac{1}{2}||\rho)}$, and observe that the error probability is upper bounded by $\exp(-\frac{\log(1/\delta)}{D_2(\frac{1}{2}||\rho)} \cdot D_2(\frac{1}{2}||\rho)) = \delta$, thus proving Lemma \ref{lem:repetition1}.
\end{proof}

\begin{proof}[Proof of Lemma \ref{lem:repetition2}]
    Similarly to the above, suppose that we individually test some item $\ttil$ times. As is well-known and was also used in \cite{Sca18}, the Chernoff bound implies for any $\zeta \in (\rho,1-\rho)$ that:
    \begin{itemize}
        \item If $j$ is defective, the probability of $\zeta \ttil$ or fewer outcomes is at most $e^{-\ttil D_2(\zeta \| 1-\rho)}$.
        \item If $j$ is non-defective, the probability of $\zeta \ttil$ or more outcomes is at most $e^{-\ttil D_2(\zeta \| \rho)}$.
    \end{itemize}
    Note that setting $\zeta = \frac{1}{2}$ recovers the bound used in the proof of Lemma \ref{lem:repetition1} above.
    
    Now suppose that this repetition testing is applied separately to each of the $k_1$ defectives and $k_0$ non-defectives.  Then, we have the following:
    \begin{itemize}
        \item {\em (Non-defectives)} By the union bound, the probability of any non-defective giving $\zeta \ttil$ or more positive outcomes is at most $k_0 e^{-\ttil D_2(\zeta \| \rho)}$.  This is at most $\delta_0$ as long as 
        \begin{equation}
            \ttil \ge \frac{\log \frac{k_0}{\delta_0}}{ D_2(\zeta \| \rho) }. \label{eq:ntil1}
        \end{equation}
        \item {\em (Defectives)} The average number of defectives returning $\zeta \ttil$ or fewer positive outcomes is upper bounded by $k_1 e^{-\ttil D_2(\zeta \| 1-\rho)}$.  Hence, by Markov's inequality, the probability of this occurring for more than $\epsilon_1 k_1$ defectives is at most $\frac{1}{\epsilon_1} e^{-\ttil D_2(\zeta \| 1-\rho)}$.  This is at most $\delta_1$ as long as
        \begin{equation}
            \ttil \ge \frac{\log \frac{1}{\epsilon_1 \delta_1}}{ D_2(\zeta \| 1 - \rho) }.  \label{eq:ntil2}
        \end{equation}
    \end{itemize}
    Recalling that $k_0 = o(k_1)$ and $\epsilon_1$ is constant, the above high-probability events ensure that the top $(1-\epsilon_1)k_1$ ranked items are all defective.   Taking $\ttil$ to equal the more stringent of the two values in \eqref{eq:ntil1} and \eqref{eq:ntil2}, this establishes Lemma \ref{lem:repetition2}.
\end{proof}

\subsection{A Multiplicative Form of the Chernoff Bound} \label{app:mult_chernoff}
 
 In generic notation, consider a binomial random variable $Z$ with $K$ trials and success probability $q = \frac{1}{\log^c N}$, for some $c > 0$.  Hence, $\EE[Z] = \frac{K}{\log^c N}$.
 
 A standard multiplicative form of the Chernoff bound states that, for any $\alpha > 0$, we have \cite[Ch.~2]{Bou13}
 \begin{align}
     \PP\big[ Z  \ge (1+\alpha)\EE[Z] \big] \le \exp\big( - \EE[Z] \big( (1+\alpha)\log(1+\alpha) - \alpha \big) \big).  \label{eq:strong_chernoff_1}
 \end{align}
 Fix any function $f(N)$ satisfying $f(N) \le o(\log^c N)$, and consider the probability $\PP\big[ Z  \ge \frac{K}{f(N)} \big]$.  Since $f(N) = o(\log^c N)$ and $\EE[Z] = \frac{K}{\log^c N}$, this amounts to choosing $\alpha = \omega(1)$ in \eqref{eq:strong_chernoff_1}; specifically, $\alpha = \frac{\log^c N}{f(N)} - 1 = \frac{\log^c N}{f(N)}(1+o(1)) = \omega(1)$.   In addition, we have the simplification $(1+\alpha)\log(1+\alpha) - \alpha = (\alpha \log \alpha)(1+o(1))$.  Hence, \eqref{eq:strong_chernoff_1} yields
 \begin{align}
     \PP\bigg[ Z  \ge \frac{K}{f(N)} \bigg] 
     &\le \exp\bigg( - \frac{K}{\log^c N} \cdot \Theta \bigg( \frac{\log^c N}{f(N)} \log \frac{\log^c N}{f(N)}  \bigg) \bigg)   \\
     &= \exp\bigg( - \frac{K}{f(N)} \cdot \Theta \bigg( \log \frac{\log^c N}{f(N)}  \bigg) \bigg).
 \end{align}
 We use this result in the following two special cases:
 \begin{itemize}
     \item If $f(N) = C$ for some constant $C > 0$, then we get $\PP[ Z \ge K/C ] \le e^{-\omega(K)}$.
     \item If $f(N) = \log\log N$, then we get $\PP[ Z \ge K / \log\log N ] \le e^{-\Omega(K)}$, since the $\Theta \big( \log \frac{\log N}{f(N)}  \big) = \Theta(\log \log N)$ term cancels with $f(N)$ (up to a constant factor).
 \end{itemize}

\bibliographystyle{IEEEtran}
\bibliography{JS_References}

\end{document}